\documentclass[runningheads]{llncs}

\usepackage{amsmath} \usepackage{complexity} \usepackage{amssymb} \usepackage{mathrsfs} \usepackage{hyperref}

\newcommand{\pdcj}{\textsc{Sorting by Prefix DCJs}}
\newcommand{\spdcj}{\textsc{Sorting by Signed Prefix DCJs}}
\newcommand{\updcj}{\textsc{Sorting by Unsigned Prefix DCJs}}

\newcommand{\upr}{\textsc{Sorting by Unsigned Prefix Reversals}}
\newcommand{\spr}{\textsc{Sorting by Signed Prefix Reversals}}
\newtheorem{observation}{Observation}

\newcommand{\decisionproblem}[3]{
\begin{center} 
\noindent\fbox{\parbox{0.97\textwidth}{
\begin{minipage}[t]{1\linewidth}
\textsc{#1}

Input: {#2}

Question: {#3}
\end{minipage}
}
}    
\end{center}
}

\usepackage{tikz}
\usetikzlibrary{fit}
\usetikzlibrary{shapes}
\tikzstyle{vertex} = [circle,fill=black!0,minimum size=4pt,inner sep=1pt]
\makeatletter
\newcommand{\gettikzxy}[3]{\tikz@scan@one@point\pgfutil@firstofone#1\relax
  \edef#2{\the\pgf@x}\edef#3{\the\pgf@y}}
\makeatother

\usepackage{xparse}
\usepackage{etoolbox} 

\newcounter{permcounter}

\NewDocumentCommand\permutationvertices{>{\SplitList{;}}m}
{
    \node[circle,fill=black,draw,inner sep=1pt] (0) at (0, 0) [label=below:0] {};
    \setcounter{permcounter}{0}
    \ProcessList{#1}{ \insertvertex }
    \stepcounter{permcounter}
    \node[circle,fill=black,draw,inner sep=1pt] (\thepermcounter) at (\thepermcounter*1.5, 0) [label=below:\thepermcounter] {};
}
\newcommand\insertvertex[1]{
    \stepcounter{permcounter};
    \node[circle,fill=black,draw,inner sep=1pt] (#1) at (\thepermcounter*1.5, 0) [label=below:#1] {};
}

\newcommand{\breakpointgraph}[1]{
\permutationvertices{#1}
\foreach [count=\q from 0] \p in {1, ..., \thepermcounter}
        \draw (\p*1.5, 0) -- (\q*1.5, 0);
\foreach [count=\q from 0] \p in {1, ..., \thepermcounter} {
        \gettikzxy{(\p)}{\px}{\py}
        \gettikzxy{(\q)}{\qx}{\qy}    
        \ifdimcomp{\qx}{<}{\px}{
            \draw[gray] (\q) to [bend left=45] (\p);
        }{
            \draw[gray] (\q) to [bend right=45] (\p);
        }
    }
}

\title{Sorting Genomes by Prefix Double-Cut-and-Joins}

\author{Guillaume Fertin\inst{1}\orcidID{0000-0002-8251-2012} \and
Géraldine Jean\inst{1}\orcidID{0000-0002-1534-2682} \and
Anthony Labarre\inst{2}\orcidID{0000-0002-9945-6774}}

\institute{Nantes Université, CNRS, LS2N, UMR 6004, F-44000 Nantes, France \and
LIGM, CNRS, Université Gustave Eiffel, F-77454 Marne-la-Vallée, France
\email{\{guillaume.fertin,geraldine.jean\}@univ-nantes.fr, anthony.labarre@univ-eiffel.fr}
}

\authorrunning{G. Fertin, G. Jean and A. Labarre}

\begin{document}

\maketitle

\begin{abstract}
In this paper, we study the problem of sorting unichromosomal linear genomes by prefix double-cut-and-joins (or DCJs) in both the signed and the unsigned settings.
Prefix DCJs cut the leftmost segment of a genome and any other segment, and recombine the severed endpoints in one of two possible ways: one of these options corresponds to a prefix reversal, which reverses the order of elements between the two cuts (as well as their signs in the signed case). 
Depending on whether we consider both options or reversals only, our main results are:
(1) 
new structural lower bounds based on the breakpoint graph for sorting by 
unsigned prefix reversals, unsigned prefix DCJs, or signed prefix DCJs; (2) a polynomial-time algorithm for sorting by signed prefix DCJs, 
thus answering an open question in~\cite{labarre-sbpbi}; 
(3) 
a 3/2-approximation for sorting by unsigned prefix DCJs, which is, to the best of our knowledge, the first sorting by {\em prefix} rearrangements problem that admits an approximation ratio strictly smaller than 2 (with the obvious exception of the polynomial-time solvable problems);
and finally, 
(4) an \FPT{} algorithm for sorting by unsigned prefix DCJs parameterised by the number of breakpoints in the genome.

\keywords{Genome Rearrangements \and Prefix Reversals \and Prefix DCJs \and Lower Bounds \and Algorithmics \and \FPT \and Approximation algorithms.}
\end{abstract}

\section{Introduction}
Genome rearrangements is a classical paradigm to study evolution between species. The rationale is to consider species by observing their genomes, which are usually represented as ordered sets of elements (the genes) that can be signed (according to gene strand when known). A genome can then evolve by changing the order of its genes, through operations called \emph{rearrangements}, which can be generally described as cutting the genome at different locations, thus forming segments, and rearranging these segments in a different fashion. 
Given two genomes, a \emph{sorting scenario} is a sequence of rearrangements transforming the first genome into the other. The length of a shortest such sequence of rearrangements is called \emph{the rearrangement distance}. Several specific rearrangements such as reversals, translocations, fissions, fusions, transpositions, and block-interchanges have been defined, and the rearrangement distance together with its corresponding sorting problem have been widely studied either by considering one unique type of rearrangement or by allowing the combination of some of them~\cite{fertin:hal-00416453}. The \emph{double-cut-and-join} (or DCJ) operation introduced by Yancopoulos \emph{et al.}~\cite{Yancopoulos2005} encompasses all the rearrangements mentioned above: it consists in cutting the genome in two different places and joining the four extremities in any possible way. A DCJ is a \emph{prefix DCJ} whenever one cut is applied to the leftmost position of the genome. The prefix restriction can be applied 
to other rearrangements such as \emph{prefix reversals}, which prefix DCJs generalise. Whereas the computational complexity of the sorting problems by unrestricted rearrangements has been thoroughly studied and pretty well characterised, there is still a lot of work to do to understand the corresponding prefix sorting problems (see Table~1 in~\cite{labarre-sbpbi} for a summary of existing results). 
Our interest in prefix rearrangements is therefore mostly theoretical: techniques that apply in the unrestricted setting do not directly apply under the prefix restriction, and new approaches are therefore needed to make progress on algorithmic issues and complexity aspects. Since DCJs generalise several other operations, we hope that the insight we gain through their study will shed light on other prefix rearrangement problems.

In this paper, we study the problem of \pdcj{} and, for the sake of simplicity, we consider the case where the source and the target genomes are unichromosomal and linear. This implies that genomes can be seen as (signed) permutations (depending on whether the gene orientation is known or not). Moreover, prefix DCJs applied to such genomes allow to exactly mimick three kinds of rearrangement: (i) a \emph{prefix reversal} when the segment between the two cuts is reversed; (ii) a  \emph{cycle extraction} when the extremities of the segment between the two cuts are joined; (iii) a \emph{cycle reincorporation} when the cut occurs in a cycle 
and the resulting linear segment is reincorporated at the beginning of the genome where the leftmost cut occurs.

Based on the study of the \emph{breakpoint graph}, we first show new structural lower bounds for the problems \updcj{} and \spdcj. Since prefix reversals are particular cases of prefix DCJs, we can extend this result to \upr{} (it has been already shown for \spr{}  in~\cite{LabarreC11}). Thanks to these preliminary results, we are able to answer an open question from~\cite{labarre-sbpbi} by proving that \spdcj{} is in \P~just like the unrestricted case~\cite{Yancopoulos2005}. However, while sorting by unsigned DCJs is \NP-hard~\cite{Chen2013}, the computational complexity of the prefix-constrained version of this problem is still unknown. We provide two additional results: a 3/2-approximation algorithm, which is, to the best of our knowledge, the first sorting by {\em prefix} rearrangements problem that admits an approximation ratio strictly smaller than 2 (with the obvious exception of the polynomial-time solvable problems); and an \FPT{} algorithm parameterised by the number of breakpoints in the genome. Due to space constraints, some of the proofs are deferred
to the Appendix.

\subsection{Permutations, genomes, and rearrangements}

We begin with the simplest models for representing organisms.

\begin{definition}
A \emph{(unsigned) permutation} of $[n]=\{1, 2, \ldots, n\}$ is a bijective application of $[n]$ onto itself. 
A \emph{signed permutation} of $\{\pm 1,\pm 2,\ldots,\pm n\}$ is a bijective application of $\{\pm 1,\pm 2,\ldots,\pm n\}$ onto itself that satisfies $\pi_{-i}=-\pi_i$. The \emph{identity permutation} is the permutation $\iota=(1\ 2\ \cdots\ n)$.
\end{definition}

We study transformations based on the following well-known operation.

\begin{definition}
A \emph{reversal} $\rho(i, j)$ with $1\leq i<j\leq n$ is a permutation that reverses the order of elements between positions $i$ and $j$:
$$
\rho(i, j)=\left(
\begin{array}{l}\renewcommand{\arraystretch}{1.3}
1\ \cdots\ i-1\ \underline{i\ \ i+1\ \cdots\ j-1\ j}\ j+1\ \cdots\ n\\
\raisebox{-.05in}{$1\ \cdots\ i-1\ j\ j-1\ \cdots\ i+1\ \ i\ j+1\ \cdots\ n$}
\end{array}
\right).
$$ 
A \emph{signed reversal} $\overline{\rho}(i, j)$ with $1\leq i\leq j\leq n$ is a signed permutation that reverses both the order and the signs of elements between positions $i$ and $j$:
$$
\overline{\rho}(i, j)=\left(
\begin{array}{ccc}\renewcommand{\arraystretch}{1.3}
1\ \cdots\ i-1 & \underline{i\ \ \ \ \ i+1\ \ \ \ \ \ \cdots\ \ \ \ \ j-1\ \ \ \ \ j} & j+1\ \cdots\ n\\
1\ \cdots\ i-1 & -j\ -(j-1)\ \cdots\ -(i+1)\ -i     & j+1\ \cdots\ n
\end{array}
\right).
$$
If $i=1$, then $\rho(i, j)$ (resp. $\overline{\rho}(i, j)$) is called a \emph{prefix (signed) reversal}.
\end{definition}

A reversal $\rho$ applied to a permutation $\pi$ transforms it into another permutation $\sigma=\pi\rho$. When the distinction matters, we mention whether objects or  transformations are signed or unsigned; otherwise, we omit those qualifiers  to lighten the presentation. The following model is a straightforward generalisation of unsigned permutations.

\begin{definition}\label{def:genome}
A \emph{genome} $G$ is a collection of vertex-disjoint paths and cycles over $\{0, 1, 2, \ldots, n+1\}$. 
It is \emph{linear} if it consists of a single path with endpoints $0$ and $n+1$. The \emph{identity genome} is the path induced by 
the sequence 
$(0, 1, 2, \ldots, n+1)$.
\end{definition}

Let us note that a genome may 
contain loops or parallel edges (see \autoref{fig:example-genomes-dcj}).

\begin{figure}[htbp]
\centering
$G$

\begin{tikzpicture}[scale=.6]
\foreach [count=\i] \v in {0,1,2,4,3,6,5} {
    \node[circle,fill=black,draw,inner sep=1pt] (\v) at (\i, 0) [label=below:\v] {};
}

\draw (0) to node[thick,cross out,draw]  {} (1) -- (2) to node[thick,cross out,draw] {} (4) -- (3) -- (6); 
\draw  (5)  to [loop above,looseness=20,out=120,in=60] (5);

\begin{scope}[xshift=-120pt,yshift=-80pt] \draw [>=stealth,->] (6, 1.5) -- (5, .5);
\foreach [count=\i] \v in {0,2,1,4,3,6,5} {
    \node[circle,fill=black,draw,inner sep=1pt] (\v) at (\i, 0) [label=below:\v] {};
}
\node at (4, -1.5) {$G_1$};
\draw[ultra thick] (0) -- (2);
\draw[ultra thick] (1) -- (4);
\draw (2) -- (1);
\draw (4) -- (3) -- (6);
\draw  (5)  to [loop above,looseness=20,out=120,in=60] (5);
\end{scope}

\begin{scope}[xshift=120pt,yshift=-80pt] \draw [>=stealth,->] (2, 1.5) -- (3, .5);
\foreach [count=\i] \v in {0,4,3,6,1,2,5} {
    \node[circle,fill=black,draw,inner sep=1pt] (\v) at (\i, 0) [label=below:\v] {};
}
\draw[ultra thick] (0) -- (4);
\draw (4) -- (3) -- (6);
\draw[ultra thick] (1) to[bend left=45] (2);
\draw (2) to[bend left=45] (1);
\draw  (5)  to [loop above,looseness=20,out=120,in=60] (5);
\node at (4, -1.5) {$G_2$};
\end{scope}

\end{tikzpicture}
\caption{Cutting edges $\{0, 1\}$ and $\{2,4\}$ from the nonlinear genome $G$ produces genome $G_1$ with a reversed segment, if we add edges $\{0, 2\}$ and $\{1,4\}$, or genome $G_2$ with an extracted cycle if we add $\{0, 4\}$ and $\{1, 2\}$ instead.}
\label{fig:example-genomes-dcj}
\end{figure}
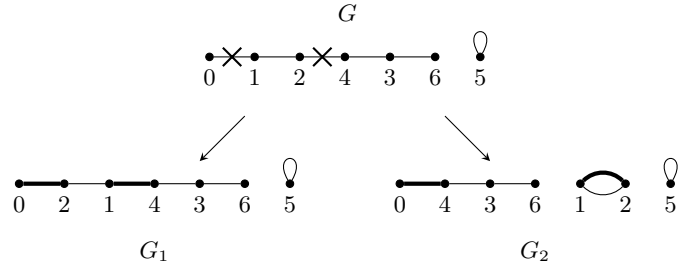

\begin{definition}
Let $e=\{u, v\}$ be an edge of a genome $G$. Then $e$ is a \emph{breakpoint} if $0\notin e$ and either $|u-v|\neq 1$, or $e$ has multiplicity two. Otherwise, $e$ is an \emph{adjacency}. The number of breakpoints of $G$ is denoted by $b(G)$.
\end{definition}

For instance, the genome with edge set $\{\{0, 4\}$, $\{4, 3\}$, $\underline{\{3, 6\}}$, $\{1, 2\}$, $\underline{\{2, 1\}}$, $\underline{\{5, 5\}}\}$ has three breakpoints (underlined). Note that permutations can be viewed as linear genomes using the following simple transformation: given a permutation $\pi$, extend it by adding two new elements $\pi_0=0$  and $\pi_{n+1}=n+1$, and build the linear genome $G_\pi$ with edge set $\{\{\pi_i, \pi_{i+1}\}\ |\ 0\le i\le n\}$. This allows us to use the notion of breakpoints on permutations as well, with the understanding that they apply to the extended permutation, and therefore $b(\pi)=b(G_\pi)$.

A reversal can be thought of as an operation that ``cuts'' (i.e., removes) two edges from a genome, then ``joins'' the severed endpoints (by adding two new edges) in such a way that the segment between the cuts is now reversed (see $G_1$ in \autoref{fig:example-genomes-dcj}). The following operation builds on that view to generalise reversals.

\begin{definition}\label{def:dcj}
\cite{Yancopoulos2005} 
Let $e=\{u, v\}\neq f=\{w, x\}$ be two edges of a genome $G$. The \emph{double-cut-and-join (or DCJ for short)} $\delta(e, f)$ applied to $G$ transforms $G$ into a genome $G'$ by replacing edges $e$ and $f$ with either $\{\{u, w\}, \{v, x\}\}$ or $\{\{u, x\}, \{v, w\}\}$. $\delta$ is a \emph{prefix DCJ} if either $0\in e$ or $0\in f$.
\end{definition}

DCJs that do not correspond to reversals extract paths from genomes and turn them into cycles (see $G_2$ in \autoref{fig:example-genomes-dcj}). Signed permutations can be generalised to signed genomes as well. The definition of a signed linear 
genome is more complicated than in the unsigned case, and is based on the following notion.

\begin{definition}\label{def:unsigned-translation}
Let $\pi$ be a signed permutation. The \emph{unsigned translation} of $\pi$ is the unsigned permutation $\pi'$ obtained by mapping $\pi_i$ onto the sequence $(2\pi_i-1,2\pi_i)$ if $\pi_i>0$, or $(2|\pi_i|,2|\pi_i|-1)$ if $\pi_i<0$, for $1\leq i\leq n$; and 
adding two new elements $\pi'_0=0$ and $\pi'_{2n+1}=2n+1$.
\end{definition}

\begin{definition}\label{def:signed-genomes}
A \emph{signed genome} $G$ is a perfect matching over the set $\{0, 1, 2, \ldots$, $2n+1\}$. $G$ 
is \emph{linear} if there exists a signed permutation $\pi$ such that $E(G)=\{\{\pi'_{2i}, \pi'_{2i+1}\}\ |\ 0\le i\le n\}$. 
The \emph{signed identity genome} is the perfect matching $\{\{2i, 2i+1\}\ |\ 0\le i\le n\}$. \end{definition}

DCJs immediately generalise to signed genomes: they may cut any pair of edges of the perfect matching, and recombine their endpoints in one of two ways.

Finally, we will be using different kinds of graphs in this work with a common notation. The \emph{length} of a cycle in a graph $G$ is the number of elements\footnote{The definition of an element will depend on the graph structure and will be explicitly stressed.} it contains, and a $k$-cycle is a cycle of length $k$: it is \emph{trivial} if $k=1$, and \emph{nontrivial} otherwise. 
We let $c(G)$ (resp. $c_1(G)$) denote the number of cycles (resp. 1-cycles) in $G$.

\subsection{Problems}\label{sec:problems}

We  study several specialised versions of the following problem. A \emph{configuration} is a permutation or  a genome, and the \emph{identity configuration} is the identity permutation or genome, depending on the type of the initial configuration.

\decisionproblem{sorting by $\Omega$}{
a configuration $G$, a number $K\in\mathbb{N}$, and a set $\Omega$ of allowed operations.}{
is there a sequence of at most $K$ operations from $\Omega$ that transforms $G$ into the identity configuration?
}

Specific choices for $\Omega$ and the model chosen for $G$ yield the following variants:\begin{itemize}
    \item \updcj{}, where $G$ is a linear genome and $\Omega$ is the set of all prefix DCJs;
    \item \spdcj{}, where $G$ is a signed linear genome and $\Omega$ is the set of all 
prefix DCJs;
    \item \upr{}, where $G$ is an unsigned permutation and $\Omega$ is the set of all prefix reversals;
    \item \spr{}, where $G$ is a signed permutation and $\Omega$ is the set of all prefix signed reversals.
\end{itemize}
 
We refer to the smallest number of operations needed to transform $G$ into the identity configuration as the \emph{$\Omega$-distance} of $G$. A specific distance is associated to each of the above problems; we use the following notation:

\begin{itemize}
    \item $pdcj(G)$ for the prefix DCJ distance of an unsigned genome $G$, and $psdcj(G)$ for its signed version;
    \item $prd(\pi)$ for the prefix reversal distance of an unsigned permutation $\pi$, and $psrd(\pi)$ for its signed version.
\end{itemize}

\section{A Generic Lower Bounding Technique}\label{sec:lb}

We present in this section a lower bounding technique which applies to both the signed and the unsigned models, and on which we will build in subsequent sections to obtain exact or approximation algorithms.

\subsection{The Signed Case}

We generalise a lower bounding technique introduced in the context of \spr{}~\cite{LabarreC11}.  It is based on the following structure.

\begin{definition}
\cite{Bafna1996} 
Given a signed permutation $\pi$, 
let $\pi'$ be its unsigned translation. 
The \emph{breakpoint graph} of $\pi$ is the undirected edge-bicoloured graph $BG(\pi)$ with ordered vertex set $(\pi'_0=0,\pi'_1, \pi'_2, \ldots, \pi'_{2n},\pi'_{2n+1}=2n+1)$ and whose edge set consists of:
\begin{itemize}
\item black edges $\{\pi'_{2i}, \pi'_{2i+1}\}$ for $0\leq i\leq n$;
\item grey edges $\{\pi'_{2i}, \pi'_{2i}+1\}$ for $0\leq i\leq n$.
\end{itemize}
\end{definition}

See \autoref{fig:breakpoint-graph-signed} for an example. Following \autoref{def:signed-genomes}, the \emph{breakpoint graph of a signed linear genome} is simply the union of that genome (which plays the role of black edges) and of the signed identity genome (which plays the role of grey edges). Breakpoint graphs are 2-regular and as such are the union of disjoint cycles whose edges alternate between both colours, thereby referred to as \emph{alternating cycles}. Black edges play the role of elements in that graph, so the \emph{length} of a cycle in a breakpoint graph is the number of black edges it contains. 

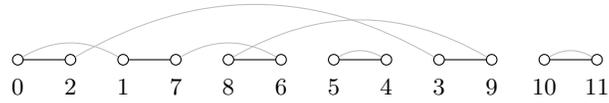
\begin{figure}[htbp]
\centering
\begin{tikzpicture}[scale=.7]
\foreach [count=\i] \color in {black,black,black,black,black,black}
        \draw[color=\color] (2*\i-2, 0) -- (2*\i-1, 0);
\foreach [count=\i from 0] \name in {0,2,1,7,8,6,5,4,3,9,10,11}
    \node[vertex] (\name) at (\i,0) [draw,circle] 
    [label=below:$\strut\name$]
    {};

\foreach \u/\v in {0/1, 2/3, 5/4, 7/6, 8/9, 10/11}
        \draw[color=gray!50] (\u) to [bend left] (\v);
\end{tikzpicture}
\caption{The breakpoint graph $BG( \pi)$ of $\pi = -1$ 4 $-3$ $-2$ 5.}
\label{fig:breakpoint-graph-signed}
\end{figure}

To bound the prefix DCJ distance, we use a connection between the effect of a DCJ on the breakpoint graph and the effect of \emph{algebraic transpositions}, or \emph{exchanges}, on the classical cycles of a permutation.

\begin{definition}\label{def:exchange}
An \emph{exchange} $\varepsilon(i, j)$ with $1\leq i<j\leq n$ is a permutation that swaps elements in positions $i$ and $j$:
$$
\varepsilon(i, j)=\left(
\begin{array}{l}\renewcommand{\arraystretch}{1.3}
1\ \cdots\ i-1\ \fbox{$i$}\ i+1\ \cdots\ j-1\ \fbox{$j$}\ j+1\ \cdots\ n\\
\raisebox{-.05in}{$1\ \cdots\ i-1\ \fbox{$j$}\ i+1\ \cdots\ j-1\ \fbox{$i$}\ j+1\ \cdots\ n$}
\end{array}
\right).
$$
If $i=1$, then $\varepsilon(i, j)$ is called a \emph{prefix exchange}.
\end{definition}

We let $\Gamma(\pi)$ denote the (directed) graph of a permutation $\pi$, with vertex set $[n]$ and  which contains an arc $(i,j)$ whenever $\pi_i=j$.
Exchanges act on two elements that belong either to the same cycle in $\Gamma(\pi)$ or to two different cycles, and therefore  $|c(\Gamma(\pi))-c(\Gamma(\pi\varepsilon(i, j)))|\le 1$. The following result allows the computation of the prefix exchange distance $ped(\pi)$ in polynomial time, and will be useful to our purposes.

\begin{theorem}\label{thm:formula-for-ped}
\cite{akers-star} For any unsigned permutation $\pi$, 
we have
$$ped(\pi)=
n+c(\Gamma(\pi))-2c_1(\Gamma(\pi))-\left\{
\begin{array}{ll}
0 & \mbox{if } \pi_1= 1, \\
2 & \mbox{otherwise}.
\end{array}
\right.$$
\end{theorem}

\begin{theorem}\label{thm:lower-bound-on-psdcj}
For any signed linear genome $G$, we have
\begin{eqnarray}\label{eqn:lower-bound-on-psdcj}
psdcj(G)\ge
n+1+c(BG(G))-2c_1(BG(G))-\left\{
\begin{array}{ll}
0 & \mbox{if } \{0, 1\}\in G, \\
2 & \mbox{otherwise}.
\end{array}
\right.
\end{eqnarray}
\end{theorem}
\begin{proof}
As observed in \cite{Yancopoulos2005}, a DCJ  acts on at most two cycles of $BG(G)$ and can therefore change the number of cycles by at most one. This analogy with the effect of exchanges on the cycles of a permutation is preserved under the prefix constraint, and the lower bound then follows from \autoref{thm:formula-for-ped}.
\qed\end{proof}

Since (prefix) signed reversals are a subset of (prefix) signed DCJs, the result below from \cite{LabarreC11} is a simple corollary of \autoref{thm:lower-bound-on-psdcj}.

\begin{theorem}\label{thm:lower-bound-on-psrd}
\cite{LabarreC11} 
For any signed permutation $\pi$, 
we have
\begin{eqnarray}\label{eqn:lower-bound-on-psrd}
 psrd(\pi)\geq
n+1+c(BG(\pi))-2c_1(BG(\pi))-\left\{
\begin{array}{ll}
0 & \mbox{if } \pi_1= 1, \\
2 & \mbox{otherwise}.
\end{array}
\right.
\end{eqnarray}
\end{theorem}

\subsection{The Unsigned Case}

We now show that our lower bounds apply to the unsigned setting as well. The definition of the breakpoint graph in the unsigned case is slightly different, but the definition of the length of a cycle remains unchanged.

\begin{definition}
\cite{Bafna1996} 
The \emph{unsigned breakpoint graph} of an unsigned permutation $\pi$ is the undirected edge-bicoloured graph $U\negthinspace BG(\pi)$ with ordered vertex set $(\pi_0=0,\pi_1, \pi_2, \ldots, \pi_{n},\pi_{n+1}=n+1)$ and whose edge set consists of:
\begin{itemize}
    \item black edges $\{\pi_{i}, \pi_{i+1}\}$ for $0\leq i\leq n$;
    \item grey edges $\{\pi_{i}, \pi_{i}+1\}$ for $0\leq i\leq n$.
\end{itemize}
\end{definition}

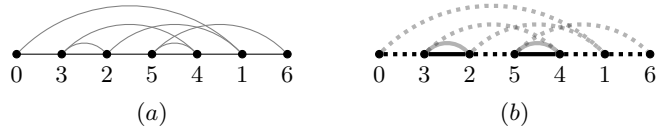
\begin{figure}
    \setlength{\tabcolsep}{12pt}
    \centering
\begin{tabular}{cc}
\begin{tikzpicture}[scale=.4]
\breakpointgraph{3;2;5;4;1}
\end{tikzpicture}
&
\begin{tikzpicture}[scale=.4]
\permutationvertices{3;2;5;4;1}
\foreach \u/\v in {3/2, 5/4} {
    \draw[ultra thick] (\u) -- (\v);
    \draw[ultra thick, opacity=.3] (\u) to [bend left=45] (\v);
}
\foreach \u/\v in {0/3,2/5,4/1,1/6} {
    \draw[ultra thick, dotted] (\u) -- (\v);
}
\foreach \u/\v in {0/1,3/4,2/1,5/6} {
    \draw[ultra thick, dotted, opacity=.3] (\u) to [bend left=45] (\v);
}
\end{tikzpicture}
\\
$(a)$ & $(b)$ \end{tabular}
\caption{$(a)$ The unsigned breakpoint graph $U\negthinspace BG(\pi)$ of $\pi =$ 3 2 5 4 1; $(b)$ an optimal decomposition of $U\negthinspace BG(\pi)$ into two trivial cycles (thick) and one 4-cycle (dotted).}
    \label{fig:ubg}
\end{figure}

\autoref{fig:ubg}$(a)$ shows an example of an unsigned breakpoint graph. Following \autoref{def:genome}, the \emph{breakpoint graph of an unsigned linear genome} is simply the union of that genome (which plays the role of black edges) and of the identity genome (which plays the role of grey edges). Vertices $0$ and $n+1$ in the unsigned breakpoint graph have degree 2, and all other vertices have degree 4. The unsigned breakpoint graph also decomposes into alternating cycles, but the decomposition is no longer unique. 
For any genome $G$ and an arbitrary decomposition $\mathscr{D}$ of  $U\negthinspace BG(G)$, 
let $c^\mathscr{D}$ (resp. $c^\mathscr{D}_1$) 
denote the number of cycles (resp. trivial cycles) of $U\negthinspace BG(G)$ in $\mathscr{D}$. We call $\mathscr{D}$  \emph{optimal} if it minimises 
$c^\mathscr{D}-2c^\mathscr{D}_1$ 
(see \autoref{fig:ubg}$(b)$). The following result characterises optimal decompositions (see Appendix for the proof).

\begin{lemma}\label{lemma:optimal-decompositions}
Let $G$ be a genome and $\mathscr{D}$ be a decomposition of $U\negthinspace BG(G)$. Then $\mathscr{D}$ is optimal iff it maximises the number of trivial cycles and minimises the number of nontrivial cycles. \end{lemma}

As a result, we obtain the following lower bound on the prefix DCJ distance, where $c^*(U\negthinspace BG(G))$ and $c^*_1(U\negthinspace BG(G))$ denote, respectively, the number of cycles and the number of $1$-cycles in an optimal decomposition of $U\negthinspace BG(G)$.

\begin{theorem}\label{thm:lower-bound-on-pdcj}
For any genome $G$, 
we have
\begin{eqnarray}\label{eqn:lower-bound-on-pdcj}
pdcj(G)
 \geq
n+1+c^*(U\negthinspace BG(G))-2c^*_1(U\negthinspace BG(G))-\left\{
\begin{array}{ll}
0 & \mbox{if } \{0, 1\} \in G, \\
2 & \mbox{otherwise}.\end{array}
\right.
\end{eqnarray}
\end{theorem}
\begin{proof}
Follows from the fact that 
DCJs 
affect the number of cycles in a decomposition by at most one, \autoref{thm:formula-for-ped}, and \autoref{lemma:optimal-decompositions}.
\qed\end{proof}

As an immediate corollary, the above lower bound is also a lower bound on $prd(\pi)$, since (prefix) reversals are a subset of (prefix) DCJs.

\begin{corollary}\label{thm:lower-bound-on-prd}
For any unsigned permutation $\pi$, 
we have
\begin{eqnarray}\label{eqn:lower-bound-on-prd}
 prd(\pi)
\geq
n+1+c^*(U\negthinspace BG(\pi))-2c^*_1(U\negthinspace BG(\pi))-\left\{
\begin{array}{ll}
0 & \mbox{if } \pi_1= 1, \\
2 & \mbox{otherwise}.\end{array}
\right.
\end{eqnarray}
\end{corollary}

We now show that an optimal decomposition can be found in polynomial time. This contrasts with the problem of finding an optimal decomposition in the case of sorting by unrestricted reversals, which was shown to be \NP-complete~\cite{caprara-sorting} (note that in that context, an optimal decomposition \emph{maximises} the number of cycles). Recall that an \emph{alternating Eulerian cycle} in a bicoloured graph $G$ is a cycle that traverses every edge of $G$ exactly once and such that the colours of every pair of consecutive edges are distinct.

\begin{corollary}\label{cor:alternating-eulerian-cycles}
\cite{Kotzig1968,DBLP:journals/algorithmica/Pevzner95} 
A bicoloured connected graph contains an alternating Eulerian cycle iff the number of incident edges of each colour is the same at every vertex.
\end{corollary}

\begin{proposition}\label{prop:optimal-decomp-poly}
There exists a polynomial-time algorithm for computing an optimal  decomposition for $U\negthinspace BG(G)$.
\end{proposition}
\begin{proof}
Straightforward: extract all trivial cycles from $U\negthinspace BG(G)$. Each connected component in the resulting graph then corresponds to a cycle (\autoref{cor:alternating-eulerian-cycles}). 
\qed\end{proof}

Finally, we note that the lower bound of \autoref{thm:lower-bound-on-pdcj} is always at least as large as the number of breakpoints (see Appendix for the proof).

\begin{proposition}\label{prop:new-lb-on-prd-better-than-breakpoints}

For any unsigned genome $G$, the lower bound from~\autoref{eqn:lower-bound-on-pdcj} is greater than or equal to $b(G)$, and the gap that separates both bounds can be arbitrarily large. 
\end{proposition}

\section{Prefix DCJs}

\subsection{Signed Prefix DCJs}
We give a polynomial-time algorithm for \spdcj{.}

\begin{theorem}
The \spdcj{} problem is in \P.\end{theorem}
\begin{proof}
We show that the lower bound of \autoref{thm:lower-bound-on-psdcj} is tight. For convenience, let 
$g(G)$ 
denote the right-hand side of \autoref{eqn:lower-bound-on-psdcj}, and let $\pi'$ denote the unsigned translation of the underlying signed permutation $\pi$ from which $G$ is obtained (recall \autoref{def:signed-genomes} and the fact that $G$ is linear):
\begin{itemize}
    \item if $\pi'_1\neq 1$: then the grey edge $\{\pi'_1, x\}$ connects by definition $\pi'_1$ to an element $x\in\{\pi'_1-1,\pi'_1+1\}$. Let $\{x, y\}$ be the black edge incident with $x$; then the prefix DCJ that replaces $\{0, \pi'_1\}$ and  $\{x, y\}$ with $\{0, y\}$ and $\{\pi'_1, x\}$ creates one or two new 1-cycles, depending on the value of $y$. Let  
$G'$ denote the resulting genome:\begin{enumerate}
        \item if $y\neq 1$, then
        \begin{align*}
         g(G')-g(G) &= n+1+c(BG(G))+1-2(c_1(BG(G))+1)-2\\
         &- (n+1+c(BG(G))-2c_1(BG(G))-2)\\
         &= -1.
        \end{align*}
        \item if $y=1$, then
        \begin{align*}
         g(G')-g(G) &= n+1+c(BG(G))+1-2(c_1(BG(G)+2))\\
         &-(n+1+c(BG(G))-2c_1(BG(G))-2)\\
         &= -1.
        \end{align*}
    \end{enumerate}
    Therefore, the value of the lower bound decreases by one in both cases.

\item otherwise, let $i$ be the smallest index such that $|\pi'_{2i-1}-\pi'_{2i}|\neq 1$. Then the prefix DCJ that replaces black edges $\{0, \pi'_1\}$ and $\{\pi'_{2i-1},\pi'_{2i}\}$ with $\{0, \pi'_{2i-1}\}$ and $\{\pi'_1, \pi'_{2i}\}$ decreases the number of $1$-cycles by $1$. Let us again use $G'$ to denote the resulting 
genome; then
        \begin{align*}
          g(G')-g(G) &= n+1+c(BG(G))-1-2(c_1(BG(G))-1)-2\\
          &-(n+1+c(BG(G))-2c_1(BG(G)))\\
         &= -1.
        \end{align*}
\end{itemize}
\qed\end{proof}

\subsection{Unsigned Prefix DCJs}

The complexity of the \updcj{} problem remains open, and we conjecture it to be \NP-complete.
Here, we prove two results, both based on the number of breakpoints. The first one is a 3/2-approximation algorithm for solving \updcj{} (\autoref{thm:approx}), the second one is a \FPT{} algorithm with respect to $b(G)$ (\autoref{thm:fpt}). 

We start with our approximation algorithm. First, observe that prefix DCJs on linear genomes may produce nonlinear genomes, but the structure of these genomes is nonetheless not arbitrary. We characterise some of their properties in the following result, which will be useful later on.

\begin{lemma}\label{lemma:linear-genome-0-and-max-always-in-same-path}
Let $G$ be a linear genome and $S$ be an arbitrary sequence of prefix DCJs that transform $G$ into a new genome $G'$. Then:
\begin{enumerate}
    \item $G'$ contains exactly one path, whose endpoints are $0$ and $n+1$;
    \item if $G'$ contains any other component, then that component is a cycle.
\end{enumerate}
\end{lemma}

\begin{proof}
By induction on $k=|S|$. If $k=0$, then the claim clearly holds. Otherwise, let $\delta$ be a prefix DCJ that cuts edges $e=\{0, v\}$ and $f=\{w, x\}$ from a genome $G''$ obtained from $G$ by $k-1$ prefix DCJs; by hypothesis, $0$ and $n+1$ are the endpoints of the only path $P$ of $G''$.
If both $e$ and $f$ belong to $P$, then $\delta$ either extracts a subpath $Q$ from $P$ that will become a cycle, or reverses a subpath $R$ of $P$; in both cases, neither $Q$ nor $R$ contains $0$ nor $n+1$, which become extremities of $P\setminus Q$ (or of the path obtained from $P$ by reversing $R$).
Otherwise, since $e=\{0, v\}$, by hypothesis $f$ belongs to a cycle, and both ways of recombining the extremities of $e$ and $f$ yield a path starting with $0$ and ending with $n+1$, preserving any other cycle of~$G''$. \qed
\end{proof}

We will need the following lower bound. 

\begin{lemma}\label{lemma:pdcj-breakpoint-lb}
For any genome $G$, we have $pdcj(G)\ge b(G)$. Moreover, if $G$ is unsorted and contains $\{0, 1\}$ and $\{1, 2\}$, then $pdcj(G)>b(G)$.
\end{lemma}

\begin{proof}
The first claim follows directly from \autoref{thm:lower-bound-on-pdcj} and \autoref{prop:new-lb-on-prd-better-than-breakpoints}.
For the second claim, 
if $G$ is unsorted and contains $\{0, 1\}$ and $\{1, 2\}$, then any new edge $\{1, y\}$ that would replace $\{0, 1\}$ would yield a breakpoint --- either because $1$ and $y$ cannot be consecutive in values or, in the event that $y=2$, because edge $\{1, 2\}$ would get multiplicity $2$ and thereby would also count as a breakpoint.
\qed\end{proof}

We are now ready to prove our upper bound on $pdcj(G)$.

\begin{lemma}\label{lemma:3-2-breakpoints}
For any linear genome $G$, we have $pdcj(G)\le \frac{3b(G)}{2}$.
\end{lemma}
\begin{proof}
Assume $G$ is not the identity genome, in which case the claim trivially holds. We have two cases to consider:
\begin{enumerate}
    \item 
  if $\{0, v\}\in G$ with $v\neq 1$, 
  then $G$ contains an element $x\in\{v-1, v+1\}$ that is not adjacent to $v$. By \autoref{lemma:linear-genome-0-and-max-always-in-same-path}, every vertex in $G$ has degree $1$ or $2$, so $x$ has a neighbour $y$ such that $\{x, y\}$ is a breakpoint (either because $|x-y|\neq 1$ or because $\{x, y\}$ has multiplicity two). 
  The prefix DCJ that replaces $\{0, v\}$ and the breakpoint $\{x, y\}$ with the adjacency $\{v, x\}$ and $\{0, y\}$ yields a genome $G'$ with $b(G')=b(G)-1$.

    \item 
    otherwise, $\{0, 1\}\in G$. If $\{1, 2\}\notin G$, then $2$ has a neighbour $y$ in $G$ such that $\{2, y\}$ is a breakpoint, in which case the prefix DCJ that replaces $\{0, 1\}$ and  $\{2, y\}$ with $\{0, y\}$ and $\{1, 2\}$ yields a genome $G'$ with $b(G')=b(G)-1$.
    
     If $\{1, 2\}\in G$, then let $k$ be the closest element to $0$ in the only path of $G$ such that the next vertex $\ell$ forms a breakpoint with $k$. Then the prefix DCJ $\delta_1$ that replaces $\{0, 1\}$ and $\{k, \ell\}$ with $\{0, \ell\}$ and $\{1, k\}$  yields a genome $G'$ which contains the cycle $(1, 2, \ldots, k)$ and with $b(G')=b(G)$. 
     Although $\delta_1$ does not reduce the number of breakpoints, $G'$ allows us to apply two subsequent operations that do:
     \begin{enumerate}
         \item since $\{0, \ell\}\in G'$ with $\ell\neq 1$, the analysis of case 1 applies and guarantees the existence of a prefix DCJ $\delta_2$ that produces a genome $G''$ with $b(G'')=b(G')-1$.
         \item $\delta_2$ replaces  $\{0, \ell\}$ 
         and breakpoint $\{a, b\}$ with $\{0, a\}$ 
         and adjacency $\{b, \ell\}$. Since $\{k, \ell\}$ was a breakpoint in $G$, we have $k<\ell-1$. Moreover, $\delta_1$ extracted from $G$ a cycle consisting of all elements in $\{1, 2, \ldots, k\}$. Therefore, the breakpoint $\{a, b\}$ cut by $\delta_2$ belongs to a component of $G''$ different from that cycle, 
         which means that $a>k>1$ and in turn implies that case 1 applies again: there exists a third prefix DCJ $\delta_3$ transforming $G''$ into a genome $G'''$ such that $b(G''')=b(G'')-1=b(G)-2$.
     \end{enumerate}
\end{enumerate}

This implies that, in the worst case, i.e. when $\{0,1\} \in G$ and $\{1, 2\} \in G$, there exists a sequence of three prefix DCJs that yields a genome $G'''$ with $b(G''')=b(G)-2$. Therefore, starting with $b(G)$ breakpoints, we can decrease this number by two using at most three prefix DCJs. Since the identity genome has no breakpoint, we conclude that $pdcj(G)\leq\frac{3b(G)}{2}$. 
\qed\end{proof}

\autoref{lemma:pdcj-breakpoint-lb} and \autoref{lemma:3-2-breakpoints} immediately imply the existence of a 3/2-ap\-prox\-i\-ma\-tion for sorting by prefix DCJs, as stated by the following theorem.

\begin{theorem}
\label{thm:approx}
The \updcj{} problem is 3/2-ap\-prox\-imable.
\end{theorem}

Note that \autoref{lemma:3-2-breakpoints} also allows us to show that our approximation algorithm is tight for an unbounded number of genomes. Incidentally, this also shows that the lower bound of  \autoref{eqn:lower-bound-on-pdcj} is optimal for an unbounded number of genomes (see Appendix for the proof).

\begin{observation}
\label{obs:lb-pdcj-tight}
There exists an unbounded number of genomes for which the algorithm described in proof of \autoref{lemma:3-2-breakpoints} is optimal.
\end{observation}

We now turn to proving that \updcj{} is \FPT{}, as stated by the following theorem.

\begin{theorem}
\label{thm:fpt}
The \updcj{} problem is \FPT{} parameterised by~$b(G)$.
\end{theorem}

\begin{proof}
The main idea is to use the search tree technique in a tree whose arity and depth are both bounded by a function of $b(G)$. For this, we will use the notion of {\em strip} in a genome $G$, which is defined as a maximal set of consecutive edges (in a path or a cycle of $G$) that contains no breakpoint. The \emph{length} of a strip is the number of elements it contains, strips of length $k$ are called $k$-strips;
1-strips are also called \emph{singletons}, and strips of length~$>2$ are called \emph{long strips}. 
We need the following result (see Appendix for the proof).

\begin{observation}
\label{obs:robust-adj}For any instance of \updcj{}, there always exists a shortest sorting sequence of prefix DCJs that never cut a long strip. 
\end{observation}

Now let us describe our search tree technique: at every iteration starting from $G$, guess in which location, among the available 2-strips and breakpoints, to operate the rightmost cut. Once this is done, guess among the two possibilities allowed by a DCJ to reconnect the genome. 
By definition, every strip is framed by breakpoints. 
Therefore, any genome $G$ has at most $b(G)$ 2-strips (recall that 
$\{0, x\}$ is never a breakpoint). Altogether, this shows that, at each iteration, the rightmost cut has to be chosen among at most $2b(G)$ possibilities. Because there are two ways to reconnect the cuts in a DCJ, the associated search tree has arity at most $4b(G)$. Moreover, its depth is at most $\frac{3b(G)}{2}$ 
since $pdcj(G)\leq \frac{3b(G)}{2}$ (\autoref{thm:approx}). Thus the above described algorithm uses a search tree whose size is a function of $b(G)$ only, which proves the result. More precisely, the overall complexity of the induced algorithm is in $O^*((4b(G))^{1.5b(G)})$.
\qed\end{proof}

\section{Conclusions and Future Work}

In this paper, we focused on the problem of sorting genomes by prefix DCJs, a problem that had not  yet been studied in its prefix-constrained version. We provided several algorithmic results for both signed and unsigned cases, including computational complexity, approximation and \FPT{} algorithms. Nevertheless, several questions remain open: while we have shown that \spdcj{} is a polynomial-time solvable problem, what about the computational complexity of \updcj{}? 
We were able to design a 3/2-approximation algorithm for the latter problem, which makes it to the best of our knowledge the first occurrence of a prefix rearrangement problem of unknown complexity where a ratio better than 2 has been obtained. 
Is 
it possible to improve it further, 
by making good use of the new lower bound introduced in \autoref{sec:lb}?
Whether or not this lower bound can help improve the 2-approximation ratios known for both \upr{} and \spr{} remains open. Finally, we have studied the case where both source and target genomes are unichromosomal and linear; it would be interesting to extend this study to a more general context where input genomes can be multichromosomal and not necessarily linear.

\bibliography{spire22-arxiv-clean}

\clearpage
\appendix
\section*{Appendix: Omitted Proofs}\label{app:omitted-proofs}

\begin{proof}[\autoref{lemma:optimal-decompositions}]
We prove both directions separately.
\begin{itemize}
    \item[$\Rightarrow$:] if 
$c^\mathscr{D}_1$ 
is not maximal, then $\mathscr{D}$ contains a nontrivial cycle from which a trivial cycle can be extracted. This yields a new decomposition $\mathscr{E}$ with 
$$
c^\mathscr{E}-2c^\mathscr{E}_1=c^\mathscr{D}+1-2(c^\mathscr{D}_1+1)<c^\mathscr{D}-2c^\mathscr{D}_1.
$$Likewise, if 
$c^\mathscr{D}-c^\mathscr{D}_1$ 
is not minimal, then $\mathscr{D}$ must contain two nontrivial cycles which can be merged into one, which decreases the quantity 
$c^\mathscr{D}-2c^\mathscr{D}_1$.

\item [$\Leftarrow$:] if $\mathscr{D}$ is not optimal, then there exists another decomposition $\mathscr{E}$ with \begin{equation}\label{eqn:simple-eqn}
c^\mathscr{D}-2c^\mathscr{D}_1> c^\mathscr{E}    -2c^\mathscr{E}_1.
\end{equation} 
We distinguish between the following three cases:
\begin{enumerate}
    \item if $c^\mathscr{D}_1=c^\mathscr{E}_1$, then $c^\mathscr{D}>c^\mathscr{E}$, so  $c^\mathscr{D}-c^\mathscr{D}_1>c^\mathscr{E}-c^\mathscr{D}_1=c^\mathscr{E}-c^\mathscr{E}_1$ and therefore $\mathscr{D}$ does not minimise the number of nontrivial cycles;
    \item if $c^\mathscr{D}=c^\mathscr{E}$, then $c^\mathscr{D}_1<c^\mathscr{E}_1$, and therefore $\mathscr{D}$ does not maximise the number of trivial cycles;
    \item if neither of the above holds, then 
    \begin{enumerate}
        \item either $c^\mathscr{D}<c^\mathscr{E}$, in which case \autoref{eqn:simple-eqn} implies $c^\mathscr{D}_1<c^\mathscr{E}_1$ and therefore $\mathscr{D}$ does not maximise the number of trivial cycles;
        
\item or  $c^\mathscr{D}>c^\mathscr{E}$, in which case either $c^\mathscr{D}_1<c^\mathscr{E}_1$,  and therefore $\mathscr{D}$ does not maximise the number of trivial cycles; or  $c^\mathscr{D}_1>c^\mathscr{E}_1$, in which case \autoref{eqn:simple-eqn} yields 
        $c^\mathscr{D}-c^\mathscr{D}_1>c^\mathscr{E}-c^\mathscr{E}_1+(c^\mathscr{D}_1-c^\mathscr{E}_1)$, which implies that $\mathscr{E}$ has fewer nontrivial cycles than $\mathscr{D}$ since $c^\mathscr{D}_1-c^\mathscr{E}_1>0$.
    \end{enumerate}
\end{enumerate}
\end{itemize}
\qed\end{proof}

\begin{proof}[\autoref{prop:new-lb-on-prd-better-than-breakpoints}]
In order to prove the inequality, we distinguish between two cases:\begin{enumerate}
\item if $\{0,1\}\not\in G$, then the lower bound from \autoref{eqn:lower-bound-on-pdcj} has value 
    \begin{equation}\label{eqn:part-i}
        n+c^*(U\negthinspace BG(G))-2c^*_1(U\negthinspace BG(G))-1.
    \end{equation} 
    Note that, by definition, each trivial cycle in $U\negthinspace BG(G)$ corresponds to an edge $\{i,i+1\}$, which is therefore an adjacency. Moreover, since $\{0,1\}\not\in G$, 0 and 1 do not form a trivial cycle, and consequently     \begin{equation}\label{eqn:part-ii}
    n=b(G)+c^*_1(U\negthinspace BG(G)).
    \end{equation} 
    Finally, since $\{0,1\}\not\in G$, $G$ is not the identity genome and we have 
    \begin{equation}\label{eqn:part-iii}
c^*(U\negthinspace BG(G))\geq c^*_1(U\negthinspace BG(G))+1.    \end{equation} 
Combining equations~\eqref{eqn:part-i}, ~\eqref{eqn:part-ii} and~\eqref{eqn:part-iii} yields the claim.

\item if $\{0,1\}\in G$, then the lower bound from \autoref{eqn:lower-bound-on-pdcj} has value 
    \begin{equation}\label{eqn:part-iv}
        n+c^*(U\negthinspace BG(G))-2c^*_1(U\negthinspace BG(G))+1.
        \end{equation} 
In that case, since 0 and 1 form a trivial cycle, we have 
\begin{equation}\label{eqn:part-v}
n+1=b(G)+c^*_1(U\negthinspace BG(G)); \end{equation} 
moreover, since $c^*(U\negthinspace BG(G))\geq c^*_1(U\negthinspace BG(G))$ always holds, combining equations~\eqref{eqn:part-iv} and~\eqref{eqn:part-v} yields the claim.
\end{enumerate}

In order to show that the gap between the lower bound of \autoref{thm:lower-bound-on-pdcj} and $b(G)$ can be arbitrarily large, consider a linear genome $G$ with $n=6p$ for an arbitrary integer $p\geq 2$. Genome $G$ corresponds to permutation $\pi$ which is the concatenation of subpermutations $\sigma_1,\sigma_2, \ldots, \sigma_p$, where for every $1\leq i\leq p$, $\sigma_i=(6i-5)\, (6i-3)\, (6i-1)\, (6i-4)\, (6i-2)\, 6i$. For instance, when $p=3$, we have $\pi=\underbrace{1\, 3\, 5\, 2\, 4\, 6}_{\sigma_1}\, \underbrace{7\, 9\, 11\, 8\, 10\, 12}_{\sigma_2}\, \underbrace{13\, 15\, 17\, 14\, 16\, 18}_{\sigma_3}$.

Formally, $G$ contains the following edges, for every $1\leq i\leq p$: $\{6i-5,6i-3\}$, $\{6i-3,6i-1\}$, $\{6i-1,6i-4\}$, $\{6i-4,6i-2\}$, $\{6i-2,6i\}$ and $\{6i,6i+1\}$; $G$ also contains edge $\{0,1\}$. 

It can be seen that $c^*_1(U\negthinspace BG(G))=p+1$, which correspond to edges $\{6i,6i+1\}$, $0\leq i\leq p$. Consequently, $b(G)=n+1-c^*_1(U\negthinspace BG(G))=5p$. Finally, $G$ has been built in such a way that, for every $1\leq i\leq p$, the elements of the interval $[6i-5;6i]$ form a cycle in $U\negthinspace BG(G)$. As a consequence, when trivial cycles are removed from $U\negthinspace BG(G)$, $p$ connected components remain, each induced by elements of $[6i-5;6i]$, $1\leq i\leq p$. Hence, as argued in proof of \autoref{prop:optimal-decomp-poly}, $G$ contains $p$ nontrivial cycles. This allows us to conclude that $c^*(U\negthinspace BG(G))=2p+1$. Altogether, since $\{0,1\}\in G$, we have that our lower bound $n+1+c^*(U\negthinspace BG(G))-2c^*_1(U\negthinspace BG(G))$ evaluates to $6p$, while as mentioned above, $b(G)=5p$, which is the sought result.
\qed\end{proof}

\begin{proof}[\autoref{obs:lb-pdcj-tight}]
Let $n=4p$ where $p\geq 2$ is any integer, and let $G$ be the 
linear genome corresponding to the following permutation $$\pi=1\, 2\, n \, (n-1)\, 5\, 6\, (n-4)\, (n-5)\, 9\, 10\, (n-8)\, (n-9)\, 13\, 14 \ldots 8\, 7\, (n-3)\, (n-2)\, 4\, 3$$ 
See \autoref{fig:example:obs:lb-pdcj-tight} for an example in the case $p=4$. 
More formally, genome $G$ contains the following edges: 
\begin{enumerate}
    \item $\{4i+1,4i+2\}$ for every $0\leq i\leq p-1$,
    \item $\{4i+2,n-4i\}$ for every $0\leq i\leq p-1$,
      \item $\{4i+3,n-4i+1\}$ for every $0\leq i\leq p-1$, 
      \item $\{4i+4,4i+3\}$ for every $0\leq i\leq p-1$, and
    \item $\{0,1\}$.
\end{enumerate}

Edge sets (1), (4) and (5) above correspond to $p+p+1$ trivial cycles, while edge sets (2) and (3) correspond to the $2p$ breakpoints that exist. Moreover, when the $2p+1$ trivial cycles are removed from $G$, in $U\negthinspace BG(G)$ we are left with a collection of $p$ 2-cycles: for each $0\leq i\leq p-1$, we have a 2-cycle whose two black edges are taken from sets (2) and (3) (namely $\{4i+2,n-4i\}$ and $\{4i+3,n-4i+1\}$), and the two grey edges are 
$\{4i+2,4i+3\}$ and $\{n-4,n-4i+1\}$.

Therefore, since $\{0,1\}\in G$, \autoref{eqn:lower-bound-on-pdcj} yields
\begin{align*}
    pdcj(G) &\ge n+1+c^*(U\negthinspace BG(G))-2c^*_1(U\negthinspace BG(G))\\
    &= 4p+1+(3p+1)-2(2p+1)\\
    &= 3p.
\end{align*}

Therefore, \autoref{eqn:lower-bound-on-pdcj} yields $pdcj(G)\geq 3p$. Moreover, since $b(G)=2p$, and by \autoref{lemma:3-2-breakpoints}, we conclude that $pdcj(G)\leq 3p$. Thus $pdcj(G)=3p$, which shows that the lower bound from \autoref{eqn:lower-bound-on-pdcj} and the algorithm described in \autoref{lemma:3-2-breakpoints} are optimal.
\qed\end{proof}

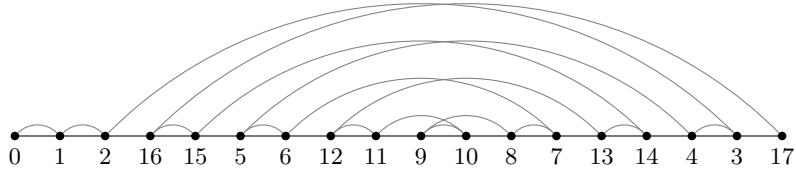
\begin{figure}[htbp]
\centering
\begin{tikzpicture}[scale=.4]
\breakpointgraph{1;2;16;15;5;6;12;11;9;10;8;7;13;14;4;3}
\end{tikzpicture}
\caption{The breakpoint graph of a genome from the family described in the proof of \autoref{obs:lb-pdcj-tight}, for $p=4$.}
\label{fig:example:obs:lb-pdcj-tight}
\end{figure}

In the following, we will observe an explicit sorting sequence for a genome, rather than the mere length of such a sequence. We refer to such sequence as a \emph{scenario}, and call it \emph{optimal} if no shorter scenario exists for the given genome. 

\begin{proof}[\autoref{obs:robust-adj}]The proof is adapted from proof of Theorem~3 in~\cite{kececioglu-sankoff}, and relies on similar arguments. However, two main differences exist in our context: first, genomes are unsigned, and thus we need to rely on {\em long} strips here, whereas in~\cite{kececioglu-sankoff} strips of length 2 or more are sufficient. Second, Theorem~3 in~\cite{kececioglu-sankoff} is concerned with reversals. The fact that we discuss prefix DCJs here implies more cases to discuss (related to DCJs only, not to the fact that operations are prefix). More precisely, observe a prefix DCJ $\delta$ in a genome $G$. Since by definition $\delta$ cuts the edge containing 0, and since by \autoref{lemma:linear-genome-0-and-max-always-in-same-path} $G$ is composed of one path $P$ (containing 0) together with cycles, there are three cases to consider: (i) the second cut of $\delta$ is in $P$ and $\delta$ is a reversal, (ii) the second cut of $\delta$ is in $P$ and $\delta$ creates a new cycle in $G$, (iii) the second cut of $\delta$ is in a cycle $C$ of $G$, and $\delta$ reincorporates $C$ in $P$.

Rather than presenting a lengthy case by case analysis, we prefer here to insist on the general arguments that ensure the proof is correct. Indeed, as it turns out, these arguments are similar whatever case we are in (cases (i), (ii) or~(iii)). 

The rationale of the proof is as follows: consider a genome, an optimal scenario $\mathcal{S}$ that sorts it, and suppose that in $\mathcal{S}$, at least one prefix DCJ cuts a long strip. Let $\delta$ be the last such prefix DCJ, let $G$ be the genome on which $\delta$ is applied, and let $S$ be the long strip that it cuts. We will show the following property, that we call $\mathcal{P}$ for convenience: it is always possible to find an alternate scenario $\mathcal{S'}$ that sorts $G$, is of same length as $\mathcal{S}$, and does not cut any long strip. In that case, it is possible to apply $\mathcal{P}$ to every prefix DCJ that cuts a long strip, from the last to the first one, and altogether, the observation is proved.

Now let us describe the alternate scenario. Recall that $\delta$ cuts strip $S$, and suppose $S$ is split into $S_A$ and $S_B$. Wlog, let us suppose that the longest substrip between $S_A$ and $S_B$ is $S_A$. Since $S$ is a long strip, we have that $S_A$ is of length at least 2, and in particular we know whether it is ascending or descending (i.e., the successive elements in $S_A$ are in increasing or decreasing order). We then replace $\delta$ by the prefix DCJ $\delta'$ that does not cut $S$, replaces $S_A$ by $S$ and deletes $S_B$ from $G$. The remaining prefix DCJs in scenario $\mathcal{S}$, among which none of them cuts a long strip, are adapted in $\mathcal{S'}$ as follows: every time $S_A$ is involved in an operation, replace it by $S$; besides, delete $S_B$ from all genomes between $G$ and the identity.

The fact that our alternate scenario $\mathcal{S'}$ also sorts $G$ relies on the following argument: in the original scenario $\mathcal{S}$, strip $S$ will be eventually grouped again, either as $S$ (if it was ascending) or as its reverse (if it was descending), so as to reach the identity genome. We just need to make sure that after the last prefix DCJ in scenario $\mathcal{S'}$, strip $S$ has the same orientation as in $\mathcal{S}$. However, this is the case since $\mathcal{S}$ and $\mathcal{S'}$ agree on $S_A$, which is of length at least 2, and thus carries information on its ``orientation" (ascending or descending). Thus strip $S$ in $\mathcal{S'}$ is reversed the same number of times as $S_A$ in $\mathcal{S}$, which is the sought property.

One final specificity needs to be taken into account. Indeed, it could happen that, in $\mathcal{S'}$, a prefix DCJ $\delta_1$ which did not cut a strip in $\mathcal{S}$ may now cut a strip. In that case, it suffices to observe that $\delta_1$ occurs strictly after $\delta$. Thus we can reproduce our previous argument to $\delta_1$, until no prefix DCJ cuts a long strip -- which always happens since at distance is 0 or 1 to the identity genome, trivially no strip is cut. 
\qed\end{proof}

\end{document}